\documentstyle[amsfonts,12pt]{article}
 
\title{Existence of Hyperbolic Calorons}
\author{Lesley Sibner\\Department of Mathematics\\Polytechnic Institute of New York University\\Brooklyn, New York 11201, USA\\ \\Robert Sibner\\Department of Mathematics\\ Brooklyn College, City University of New York\\Brooklyn,
New York 11210, USA \\ \\Yisong Yang\footnote{Corresponding author. Email address: yisongyang@nyu.edu}\\Institute of Contemporary Mathematics\\School of Mathematics and Statistics\\Henan University\\Kaifeng, Henan 475000, PR China\\NYU--ECNU
Institute of Mathematical Sciences\\New York University - Shanghai\\3663 North Zhongshan Road, Shanghai 200062, PR China}
\date{}
\newcommand{\bfR}{{\Bbb R}}

\newcommand{\bfH}{{\Bbb H}}

\newtheorem{oldtheorem}{Theorem}[section]
\newtheorem{oldassertion}[oldtheorem]{Assertion}
\newtheorem{oldproposition}[oldtheorem]{Proposition}

\newtheorem{oldlemma}[oldtheorem]{Lemma}
\newtheorem{olddefinition}[oldtheorem]{Definition}
\newtheorem{oldclaim}[oldtheorem]{Claim}
\newtheorem{oldcorollary}[oldtheorem]{Corollary}

\newenvironment{theorem}{\begin{oldtheorem}$\!\!\!${\bf.}}{\end{oldtheorem}}

\newenvironment{lemma}{\begin{oldlemma}$\!\!\!${\bf.}}{\end{oldlemma}}

\newbox\qedbox
\newenvironment{proof}{\smallskip\noindent{\bf Proof.}\hskip \labelsep}%
                        {\hfill\penalty10000\copy\qedbox\par\medskip}

\newcommand{\dd}{\mbox{d}}
\newcommand{\ee}{\end{equation}}
\newcommand{\be}{\begin{equation}}\newcommand{\bea}{\begin{eqnarray}}
\newcommand{\eea}{\end{eqnarray}}
\newcommand{\ii}{\mbox{i}}\newcommand{\e}{\mbox{e}}
\newcommand{\pa}{\partial}\newcommand{\Om}{\Omega}
\newcommand{\vep}{\varepsilon}

\newcommand{\nn}{\nonumber}

\begin{document}
\maketitle
\begin{abstract} Recent work of Harland shows that the $SO(3)$-symmetric, dimensionally-reduced, charge-$N$ self-dual Yang--Mills calorons on the 
hyperbolic space $\bfH^3\times S^1$ may be
obtained through constructing $N$-vortex solutions of an Abelian Higgs model as in the study of Witten on
multiple instantons. In this paper we establish the existence of such minimal action
charge-$N$ calorons by constructing arbitrarily prescribed $N$-vortex solutions of the Witten type equations.
\end{abstract}

\section{Introduction}
\setcounter{equation}{0}

Instantons \cite{Actor} are topological solitons of the zero-temperature Yang--Mills equations in the Euclidean space $\bfR^4$ obtained from the $(3+1)$-dimensional Minkowski spacetime so that the
time axis is made imaginary by a Wick rotation. These classical solutions give leading-order contributions in the partition function calculation and describe tunneling
between various ground states in quantum field theory formalism \cite{Jackiw}. At finite temperature, $T>0$, one needs to compactify the Euclidean time, $t$, which means that
instantons become time-periodic \cite{HS1} so that $t$ is confined within the temporal cell
\be 
0\leq t\leq \beta=\frac1{kT},
\ee
 where $k$ is the Boltzmann constant, so that the normalized partition function assumes an apparently asymmetric form
\be \label{Z}
Z=\int {\cal D}A\exp\left(-\int_0^\beta\int_{\bfR^3}L(A)\,\dd x\dd t\right),
\ee
where $L(A)$ is the action density of the gauge field $A$ and ${\cal D}A$ denotes the path-integral measure. As in the situation of zero-temperature instantons, finite-action $t$-periodic
gauge field solutions can be stratified  \cite{HS1} by homotopy classes defined by maps from $S^2\times S^1$ into $S^3$
 and such finite-temperature instantons have been explicitly constructed by Harrington and Shepard \cite{HS2}, and called calorons, which approach the zero-temperature instantons in
the limit $\beta\to\infty$ so that the asymmetry between the spatial and temporal coordinates, $x$ and $t$, in the partition function (\ref{Z}) disappears.
Motivated by the work on hyperbolic monopoles \cite{Atiyah,Ch,Nash} in the extreme curvature limit \cite{Nor} in connection to the Euclidean monopoles, Harland 
\cite{Harland} carried out a study of hyperbolic calorons and showed that, within Witten's $SO(3)$-symmetric dimensionally reduced ansatz \cite{Witten}, hyperbolic calorons may be
obtained through constructing multiple vortex solutions of an Abelian Higgs Bogomol'nyi system over a cylindrical stripe. Specifically, a unit charge  caloron is presented
and large $\beta$ period and large hyperbolic space curvature limits are discussed \cite{Harland}. However, the existence of a general charge caloron solution remains unsolved.

Technically, the difficulty lies in the fact that the reduced governing equation is defined over an infinite stripe domain in $\bfR^2$ equipped with
an exponentially curved metric. As in the study of Witten \cite{Witten}, the vortex equation may be reduced to a Liouville equation which is known to be integrable. However, the sign of the
nonlinearity only allows local solvability of the equation \cite{Sattinger} and periodicity also introduces additional complexity \cite{Olesen}. The purpose of this paper is to apply the
method of nonlinear functional analysis to establish the existence of arbitrarily prescribed multiple vortex solutions of the Abelian Higgs Bogomol'nyi system derived in the work of Harland
\cite{Harland}. Thus it follows that the existence of an arbitrary-charge hyperbolic caloron \cite{Harland} is obtained.

In the next section, we follow \cite{Harland} to introduce the hyperbolic caloron problem. In particular, we recall the Bogomol'nyi equations of Harland \cite{Harland} defined over a  cylindrical stripe,
similar to Witten's equations for the dimensionally reduced $SO(3)$-symmetric instanton problem \cite{Witten}. In Section 3, we recall the governing elliptic equation in terms of the coordinate
variables.
In Section 4, we prove the existence of solutions by using a variational approach
and a sub- and supersolution argument similar to the method used in constructing the Witten type instanton solutions in $4m$ dimensions \cite{SSY,STY} systematically 
developed by Tchrakian \cite{T1,T2,T3}.
Note that, unlike the problem in \cite{SSY,STY}, the exponential decay of the curved metric and
finite periodicity make it hard to gain precise information of a solution at infinity. Although we know that the solution remains bounded, we do not know whether it
has a definite asymptotic value at infinity. In particular, we do not know its uniqueness. In Section 5, we deduce some suitable decay estimates 
near the boundaries of the stripe domain for the solution obtained
which allow us to compute the associated topological charge realized as the
total magnetic flux, or the second Chern class, explicitly. In Section 6, we obtain decay estimates for the gradient of the solution obtained.
These estimates make the formal Bogomol'nyi reduction legitimate and lead us to the conclusion 
in Section 7 that the action of a charge $N$ caloron 
represented by the gauge field $A$ carries the anticipated
minimum action, $S(A)=2\pi^2 N$, in normalized units. 
 
\section{Hyperbolic calorons and vortices}
\setcounter{equation}{0}

Following Harland \cite{Harland}, the radial coordinate $R$ of a point $(x^1,x^2,x^3)$ in the hyperbolic ball $\bfH^3$ is given by
$R^2=(x^1)^2+(x^2)^2+(x^3)^2$ which is confined in the interval $0\leq R<S$ with $S>0$ the scalar curvature of $\bfH^3$. The temporal coordinate $x^0=t$ is 
of period $\beta>0$ and parametrizes $S^1$.  Using $\dd\Om^2$ to denote the metric on the standard 2-sphere and introducing the new variable 
\be 
r=\frac S2 \tanh^{-1}\left(\frac RS\right),\quad 0\leq r<\infty,
\ee
the metric for $\bfH^3\times S^1$ is given by
\be \label{metric}
\dd s^2 =\dd t^2+\dd r^2 +\Xi^2\dd\Om^2,
\ee
where the conformal factor $\Xi>0$ is defined by the function
\be\label{4}
{\Xi}={\Xi}(r)=\frac S2\sinh \bigg(\frac{2r}S\bigg),\quad r>0.
\ee

Let $A$ be an $su(2)$-valued connection 1-form or the gauge field, $A=A_\mu\dd x^\mu$, with the associated curvature 2-form $F_A=\dd A+[A,A]$, and $*$ the Hodge
star operator induced from the metric (\ref{metric}) over the manifold $M=\bfH^3\times S^1$. The Yang--Mills action, $S(A)$, of $A$, is given in the standard form
\be\label{action}
S(A)=-\int_M \mbox{Tr} (F_A\wedge *F_A),
\ee
accompanied with the topological charge 
\be\label{topo}
Q(A)=c_2(A)=-\frac1{8\pi^2}\int_M \mbox{Tr}(F_A\wedge F_A),
\ee
expressed formally as the second Chern class or the first Pontryagin class. Finite-action condition implies \cite{Harland} that the field decays appropriately at the boundary of $M$ so that one can recognize the topological
lower bound
\be \label{lower}
S(A)\geq 2\pi^2 |Q(A)|.
\ee
A caloron has a positive charge, $Q(A)>0$,  saturates the lower bound (\ref{lower}), $S(A)=2\pi^2 Q(A)$, and satisfies the
self-dual equation
\be\label{YM}
F_A=*F_A. 
\ee

As in \cite{Harland}, we are interested in the vanishing holonomy situation where
$Q(A)$ is an integer. To proceed further, we follow \cite{Harland,Landweber} to represent the $SO(3)$-symmetric gauge field $A$ in terms of an Abelian gauge field $a=a_t\dd t+a_r\dd r$ and a complex
scalar Higgs field $\phi=\phi_1-\ii\phi$ as
\be \label{AS}
A=-\frac12(qa+\phi_1\dd q+[\phi_2+1]q\dd q),
\ee
where $q=x^j\sigma^j/R$ with $\sigma^j$ ($j=1,2,3$) denoting the Pauli spin matrices. Thus, in terms of the reduced Abelian curvature $F_a=\dd a$ and connection $\dd_a\phi=\dd\phi+\ii a\phi$, the 
Yang--Mills action over $M=\bfH^3\times S^1$ boils down into an Abelian Higgs action over the cylindrical stripe 
\be \label{strip}
{\cal M}=\{(r,t)\,|\,0<r<\infty,\,0\leq t\leq\beta\}, 
\ee
equipped with the metric
\be \label{dell}
\dd\ell^2=\frac1{\Xi^2}(\dd t^2+\dd r^2),
\ee
of the form
\be \label{action-phi-a}
S(A)=S(\phi,a)=\frac\pi2\int_{\cal M}\left(F_a\wedge * F_a+2\dd_a\phi\wedge *\overline{\dd_a\phi}+*(1-|\phi|^2)^2\right),
\ee
where now $*$ is understood to be the Hodge dual with respect to the metric (\ref{dell}) on $\cal M$. Furthermore, the topological charge $Q(A)$ becomes the first Chern number
\be 
Q(A)=\frac1{2\pi}\int_{\cal M} F_a=c_1(a).
\ee
Accordingly, using the method of Bogomol'nyi \cite{Bo,JT}, it can formally be shown that the lower bound stated in (\ref{lower}) is attained if the pair $(\phi,a)$ satisfies
the self-dual vortex equations over $\cal M$,
\bea 
\dd_a\phi+\ii *\dd_a\phi&=&0,\label{v1}\\
\mbox{$*$} F_a&=&1-|\phi|^2.\label{v2}
\eea
These equations can also be reduced from the original Yang--Mills equation (\ref{YM}) via the $SO(3)$-symmetric ansatz (\ref{AS}) as described in \cite{Harland}.
In view of such a connection, our main existence theorem for calorons may be stated as follows.

\begin{theorem} For any integer $N\geq1$, the self-dual Yang--Mills equation (\ref{YM}) over the hyperbolic space $\bfH^3\times S^1$ has a
$2N$-parameter family of smooth solutions, say $\{A\}$, realizing the prescribed topological invariant $Q(A)=c_2(A)=N$ so that the action (\ref{action})
saturates the topological lower bound stated in (\ref{lower}), $S(A)=2\pi^2 Q(A)=2\pi^2 N$. In fact, such solutions may be obtained by constructing multivortex solutions
of (\ref{v1}) and (\ref{v2}) representing $N$ vortices realized as zeros of the complex Higgs field $\phi$ over a cylindrical 2-surface $\cal M$ defined in (\ref{strip})
and equipped with the metric (\ref{dell}).
\end{theorem}

In the subsequent sections, we aim at solving the coupled equations (\ref{v1}) and (\ref{v2}), which belong to a category of gauge field equations over Riemann surfaces
known as Hitchin's equations \cite{Hitchin}.

\section{Elliptic governing equation}
\setcounter{equation}{0}

It will be convenient to rewrite (\ref{v1}) and (\ref{v2}) in terms of the $(r,t)$-coordinates as in \cite{Harland}. Thus these equations become
\bea 
D_r \phi +\ii D_t \phi&=&0,\label{1}\\
{\Xi}^2 (\pa_t a_r-\pa_r a_t)&=& 1-|\phi|^2,\label{2}
\eea
where
$
D_r\phi=\pa_r\phi+\ii a_r\phi$ and $ D_t\phi=
\pa_t\phi+\ii a_t\phi
$
 are the gauge-covariant derivatives of $\phi$ and the field configurations are all $\beta$-periodic in the variable $t$.

Use complex variables to represent the equations with the convention
\be\label{5}
z=r+\ii t,\quad a=a_r+\ii a_t,\quad \pa_z=\pa=\frac12(\pa_r-\ii\pa_t),\quad \pa_{\overline{z}}=\overline{\pa}=\frac12(\pa_r+\ii\pa_t).
\ee
Then (\ref{1}) takes the form
\be\label{6}
\overline{\pa}\phi=-\frac\ii2 a\phi,
\ee
so that, away from the zeros of $\phi$, we have $a=2\ii\overline{\pa}\ln\phi$. On the other hand, noting that
\be\label{7}
\pa a -\overline{\pa}\,\overline{a}=\ii (\pa_r a_t-\pa_t a_r),
\ee
we see that the nontrivial $rt$-component of the curvature of $a$ may be represented as
\be\label{8}
F_{rt}=(F_a)_{rt}=\pa_r a_t-\pa_t a_r=2\pa\overline{\pa}\ln|\phi|^2=\frac12(\pa_r^2+\pa_t^2)\ln|\phi|^2=\frac12\Delta\ln|\phi|^2.
\ee

In view of (\ref{8}), we can rewrite (\ref{2}) away from the zeros of $\phi$ as
\be\label{9}
\frac12 {\Xi}^2\Delta \ln |\phi|^2 =|\phi|^2 -1.
\ee

It is well known that the zeros of $\phi$ are discrete and have integer multiplicities. Let these zeros be
\be\label{10}
p_1,p_2,\cdots,p_N,
\ee
where and in the sequel, a zero of multiplicity $m$ appears in the list (\ref{10}) $m$ times. Set $u=\ln|\phi|^2$. Then, over the full space $\cal M$, (\ref{9}) becomes
\be\label{11}
\Delta u=\frac2{{\Xi}^2}(\e^u-1)+4\pi\sum_{j=1}^N\delta_{p_j}.
\ee
We are to look for a solution of (\ref{11}) satisfying the boundary condition
\be\label{12}
u(r,t) \to 0\mbox{ as } r\to0;\quad u(r,t)\mbox{ is of period $\beta$ in the variable $t$};
\ee
$u$ stays bounded over $\cal M$.

Conversely, the Higgs field $\phi$ has the amplitude $|\phi|^2=\e^u$ in terms of a solution $u$ of (\ref{11}) and the Abelian gauge field $a_t, a_r$ may be constructed
from utilizing (\ref{1}) to give us
\be 
a_t=\mbox{Re}\{2\ii\overline{\pa} \ln u\},\quad a_r=\mbox{Im}\{2\ii\overline{\pa}\ln u\},
\ee
which allows us to find the useful relation
\be\label{DD}
|D_t\phi|^2+|D_r\phi|^2=\frac12\e^u|\nabla u|^2.
\ee

Returning to the equation (\ref{11}), staying away from the points $p_1,p_2,\cdots, p_N$, and using the translation
\be 
u=2\ln \Xi +v,
\ee
we see that the function satisfies the Liouville equation
\be \label{Liouville}
\Delta v=2\e^v,
\ee
which is integrable. However, such an integrability is only local because (\ref{Liouville}) is known \cite{Sattinger}
to have no entire solution
over $\bfR^2$, although our problem requires that the solution be of period $\beta$ in its $t$ variable. In doubly periodic
case, the solutions to the Liouville equation are considered by Olesen \cite{Olesen,Olesen2} in the context of nonrelativistic Chern--Simons vortices and electroweak vortices over periodic lattices where one needs to use
the elliptic functions \cite{Ak,Chan,Lang} of Weierstrass since the holomorphic functions representing solutions of (\ref{Liouville})
are periodic. In our situation here, complication comes from both the periodicity of the solution 
$v$ of (\ref{Liouville}) in the $t$ variable
and unboundedness of $v$ in the $r$ variable as $r\to0$ and $r\to\infty$, respectively, as a consequence of the form of the background function
$\Xi$ given in (\ref{4}). Due to these difficulties, we choose to use analytic methods to study (\ref{11}) directly,
rather than treating it as an integrable equation, so that the desired boundary conditions and arbitrarily prescribed
distribution of vortices, as well as their relations to the calculation of
topological charges and minimal actions, can all be realized directly and readily.
 
\section{Construction of solution to the vortex equation}
\setcounter{equation}{0}

For $x=(r,t)$, set
\be 
U_0(x)=-\sum_{j=1}^N \ln(1+|x-p_j|^{-2}),\quad x\in\bfR^2.
\ee
Then $U_0<0$. Let $\eta(x)$ be a smooth cut-off function such that $0\leq\eta\leq1$, $\eta$ is of compact support in the rectangle $(0,\infty)\times(0,\beta)$, and
\be 
\eta=1\quad\mbox{in a neighborhood of }\{p_1,p_2,\cdots, p_N\}.
\ee
Then $u_0=\eta U_0$ is of compact support in $(0,\infty)\times(0,\beta)$, $u_0\leq0$,  and
\be\label{2.1}
\Delta u_0=4\pi\sum_{j=1}^N\delta_{p_j}-g(r,t)
\ee
with $g\in C_0^\infty((0,\infty)\times(0,\beta))$ (set of smooth functions of compact supports).

Rewrite $u$ in (\ref{11}) as $u=u_0+v$. Then we have
\be\label{2.2}
\Delta v=\frac2{{\Xi}^2}(\e^{u_0+v}-1)+g(r,t)
\ee
which is so defined that we are interested in solution depending on the radial variable $r>0$ and the Euclidean time variable $t$ of period $\beta$.

\begin{lemma}\label{lemma0} The function $v^+=-u_0$ is an upper solution of (\ref{2.2}).
\end{lemma}

\begin{proof} The function $v^+$ is clearly of period $\beta$ in the variable $t$. Besides, in sense of distribution, we have from (\ref{2.1}) the inequality
\be 
\Delta v^+=\Delta (-u_0)\leq\frac2{{\Xi}^2} (\e^{u_0+v^+}-1)+g(r,t)
\ee
as desired.
\end{proof}

We next construct a lower solution of (\ref{2.2}). For this purpose, we define
\be\label{def}
G(r)=\max\{\max\{g(r,t),0\}\,|\, 0\leq t\leq\beta\}.
\ee

We consider the boundary value problem
\bea
w_{rr}&=&\frac2{{\Xi}^2} (\e^w-1)+G(r),\quad r>0,\label{ODE1}\\
w(0)&=&0,\quad w(\infty)=w_\infty,\label{BC1}
\eea
where $w_\infty\leq0$ is an undetermined constant.

The singular nature of the equation (\ref{ODE1}) does not allow us to approach it directly. Instead, we consider the approximate boundary value problem
\bea
w_{rr}&=&\frac2{{\Xi}^2} (\e^w-1)+G(r),\quad \vep_n<r<K_n,\label{ODE2}\\
w(\vep_n)&=&0,\quad w(K_n)=0,\label{BC2}
\eea
for $n=1,2,\cdots$. Here $\{\vep_n\}$ and $\{K_n\}$ are monotone sequences of positive numbers with $\vep_n<K_n$, $\mbox{supp}(G)\subset (\vep_n, K_n)$, $n=1,2,\cdots$, and
\be 
\lim_{n\to\infty}\vep_n=0,\quad\lim_{n\to\infty}K_n=\infty.
\ee

\begin{lemma}\label{lemma2.1} The boundary value problem consisting of (\ref{ODE2}) and (\ref{BC2})  has a unique solution which may be
obtained by minimizing the functional
\be\label{2.4}
I_n(w)=\int_{\vep_n}^{K_n} \bigg\{\frac12(w_r)^2+\frac2{{\Xi}^2}({\rm{\e}}^{w}-1-w)+G(r) w\bigg\}\,{\rm\dd} r
\ee
in the space $W^{1,2}_0(\vep_n,K_n)$.
\end{lemma}

\begin{proof} Since ${\Xi}(r)\geq {\Xi}(\vep_n)>0$ for $r\in(\vep_n,K_n)$ and $\e^w-1-w\geq 0$, we can use the Schwarz inequality and the Poincar\'{e} inequality to derive easily the coerciveness
of the functional $I_n$, namely, $I_n(w)\geq C_1 \|w\|_{W^{1,2}_0(\vep_n,K_n)}-C_2$ for some constants $C_1, C_2>0$. Hence the existence of critical point of $I_n$ in $W^{1,2}_0(\vep_n,K_n)$
follows which solves the boundary value problem (\ref{ODE2})--(\ref{BC2}). The uniqueness can be proved using a maximum principle argument in (\ref{ODE2})--(\ref{BC2}).
\end{proof}

In order to pass to the $n\to\infty$ limit, we need the following monotonicity results. For notational convenience, we will use $f_r$ and $f'$ interchangeably to denote the
derivative of a function $f$ with respect to the radial variable $r$.

\begin{lemma}\label{lemma2.2} Let $w_n$ be the unique solution of the boundary value problem (\ref{ODE2})--(\ref{BC2})
obtained in Lemma \ref{lemma2.1}. Then there hold the monotonicity relation
\be\label{2.6}
I_n(w_n)\geq I_{n+1}(w_{n+1}),\quad n\geq 1,
\ee
and the uniform coerciveness lower bound
\be\label{2.7}
I_n(w_n)\geq C_1\int_{\vep_n}^{K_n}(w_n'(r))^2\,{\rm\dd} x-C_2,\quad n\geq1,
\ee
where $C_1, C_2>0$ are independent of $n$. Furthermore, the sequence $\{w_n\}$ is monotone-ordered according to
\be\label{2.8}
0>w_n>w_{n+1},\quad \vep_n<r<K_n,\quad n\geq1.
\ee
\end{lemma}
\begin{proof}
First we recall that for $n\geq1$ the function $w_n$ is the unique minimizer of the functional $I_n$ in $W^{1,2}_0(\vep_n,K_n)$.
Set $w_n=0$ for $r<\vep_n$ and $r>K_n$. Then $w_n\in W^{1,2}_0(\vep_{n+1},K_{n+1})$ and $I_{n+1}(w_n)= I_n(w_n)$. However, $w_{n+1}$
is the global minimizer of $I_{n+1}$ in $W^{1,2}_0(\vep_{n+1},K_{n+1})$. Hence $I_{n+1}(w_{n+1})\leq I_{n+1}(w_n)$ and (\ref{2.6}) is established.

Let $f(r)$ be a function so that $f(r)=0$ when $r>0$ is sufficiently small or large. Then, an integration by
parts gives us
\be\label{2.9}
\int_0^\infty\frac1{r^2} f^2(r)\,\dd r=2\int_0^\infty \frac1r f(r) f'(r)\,\dd r.
\ee
Thus, by the Schwarz inequality, we have
\be\label{2.10}
\int_0^\infty\frac1{r^2} f^2 (r)\,\dd r\leq 4\int_0^\infty (f_r)^2\,\dd r.
\ee
Using the elementary inequality $\e^w-1-w\geq0$ again and (\ref{2.10}), we have
\bea\label{2.11}
I_n(w_n)&\geq&\frac12\int_{\vep_n}^{K_n}( w_n')^2\,\dd r -\bigg(\int_{\vep_n}^{K_n} r^2 G^2\,\dd r \bigg)^{1/2}
\bigg(\int_{\vep_n}^{K_n}\frac{w^2_n}{r^2}\,\dd r\bigg)^{1/2}\nn\\
&\geq&\frac14\int_{\vep_n}^{K_n}(w_n')^2\,\dd r  -4\int_{\vep_n}^{K_n}r^2 G^2\,\dd r,
\eea
which gives us (\ref{2.7}).

Finally, applying the maximum principle and the condition $G(r)\geq0$ in (\ref{ODE2})--(\ref{BC2}),
 we see that $w_n<0$ in $(\vep_n,K_n)$. In particular, $w_{n+1}<0$ on $[\vep_n,K_n]$.  Now in $(\vep_n,K_n)$ the function $w_{n+1}-w_n$ satisfies
\be\label{2.12}
(w_{n+1}-w_n)_{rr}=\frac2{{\Xi}^2}\e^{\xi_n}(w_{n+1}-w_n) \mbox{ where $\xi_n$ lies between $w_n$ and $w_{n+1}$},
\ee
and $(w_{n+1}-w_n)(r)<0$ for $r=\vep_n$ and $r=K_n$. Applying the maximum principle to (\ref{2.12}) gives us $w_{n+1}<w_n$ in $(\vep_n,K_n)$ or (\ref{2.8}).
\end{proof}

\begin{lemma} \label{lemma2.4} The sequence $\{w_n\}$
constructed in Lemma \ref{lemma2.2} is weakly convergent in $W^{1,2}_{\rm\mbox{loc}}(0,\infty)$. The so-obtained weak limit, say $w$, is a classical solution of the equation
(\ref{ODE1}). In fact, the convergence $w_n\to w$ ($n\to\infty$) may be achieved in any $C^k[a,b]$ topology for arbitrary $0<a<b<\infty$. In particular, we have $w\leq0$ everywhere.
\end{lemma}
\begin{proof}
From (\ref{2.6}) and (\ref{2.7}), we see that there is an absolute constant $C>0$ such that
\be\label{2.13}
\sup_n\bigg\{\int_{\vep_n}^{K_n}(w_n'(r))^2\,\dd r\bigg\}\leq C.
\ee
From (\ref{2.10}) and (\ref{2.13}), we see that $\{w_n\}$ is bounded in $W^{1,2}(a,b)$ for arbitrary $0<a<b<\infty$. In view of the monotonicity (\ref{2.8}), we conclude that $\{w_n\}$ is weakly
convergent in $W^{1,2}(a,b)$. Using extension, we can find a function $w\in W^{1,2}_{\rm\mbox{loc}}(0,\infty)$ such that  $\{w_n\}$
converges to $w$ in $W^{1,2}(a,b)$ for any $0<a<b<\infty$.

Choose $n_0\geq1$ such that $(a,b)\subset (\vep_n,K_n)$ when $n\geq n_0$. Thus, for any test function
$\xi\in C_0^1(a,b)$, we have
\be\label{2.14}
\int_a^b \bigg\{ w_n' \xi'+\frac2{{\Xi}^2}(\e^{w_n}-1)\xi +G(r)\xi\bigg\}\,\dd r=0,\quad n\geq n_0.
\ee
Using the weak convergence of $\{w_n\}$ in $W^{1,2}(a,b)$, we see that
$\{w_n\}$ is convergent in $C[a,b]$ as well. 
 We can take $n\to\infty$ in (\ref{2.14}) to show that $w$ is a weak solution
of (\ref{ODE1}) over $(a,b)$. Since $(a,b)$ is arbitrary, we see that $w$ is a weak solution of (\ref{ODE1}) over the full domain $r>0$. Standard elliptic theory then implies that
$w$ is a classical solution of (\ref{ODE1}).

The convergence in any $C^k[a,b]$ topology for arbitrary $0<a<b<\infty$ follows from applying elliptic estimates in the equation
\be 
w_n''=\frac2{{\Xi}^2}(\e^{w_n}-1)+G(r)
\ee
and the property $w_n\to w$ in $C[a,b]$ as $n\to\infty$.
\end{proof}

\begin{lemma} \label{lemma2.5}
Let $w$ be the solution of (\ref{ODE1}) obtained in Lemma \ref{lemma2.4}. Then it satisfies the boundary condition (\ref{BC1}) for some unique number $w_\infty\leq0$
so that for arbitrarily small $\vep>0$ we have $w(r)={\rm\mbox{O}}(r^{2-\vep})$ as $r\to0$ and $w(r)-w_\infty={\rm\mbox{O}}({\rm\e}^{-4r/S})$ as $r\to\infty$.
\end{lemma}

\begin{proof} Let $\{w_n\}$ be the sequence of solutions of (\ref{ODE2})--(\ref{BC2}) obtained in Lemmas \ref{lemma2.1} and \ref{lemma2.2}. Then for any
$r\in(\vep_n,K_n)$, we have by the Schwarz inequality and (\ref{2.13}) the uniform bound
\be 
|w_n(r)|\leq \int_{\vep_n}^r|w_n'(\rho)|\,\dd \rho\leq r^{1/2}\bigg(\int_{\vep_n}^{K_n}(w_n'(r))^2\,\dd r\bigg)^{1/2}\leq C r^{1/2},
\ee
where $C>0$ is independent of $n$ and $r$. Hence $w(r)=\mbox{O}(r^{1/2})$ when $r\to0$ which is a crude preliminary estimate. To improve it, we consider a comparison function
\be \label{W}
W(r)=Cr^{2-\vep},\quad r>0,\quad C>0, \quad\vep\in(0,1),
\ee
and set $U=w+W$. Choose $r_0>0$ small such that $G(r)=0$ for $0<r<r_0$. In view of (\ref{ODE1}), (\ref{W}), 
\be\label{V0} 
{\Xi}(r)=r\cosh(\frac{2\xi(r)}S), \quad \xi(r)\in(0,r), 
\ee
and $w(0)=0$, we have
\be\label{U}
U_{rr}=\frac2{{\Xi}^2} (\e^w-1)+(2-\vep)(1-\vep) r^{-2} W
= r^{-2} (K(r)w+(2-\vep)(1-\vep) W),
\ee
where $0<r<r_0$ and $K(r)\to2$ as $r\to0$. Hence, when $r_0$ is small, we have $K(r)>(2-\vep)(1-\vep)$ for $r\in (0,r_0)$. Inserting this condition into (\ref{U}), we have 
\be \label{U1}
U_{rr}<r^{-2}K(r)U,\quad 0<r<r_0.
\ee
Using $U(0)=w(0)+W(0)=0$ and assuming $C>0$ in (\ref{W}) is large enough so that $U(r_0)=w(r_0)+W(r_0)>0$. Applying these in (\ref{U1}), we get $U(r)>0, r\in(0,r_0)$. That is,
we have obtained the estimate
\be 
0\geq w(r)>-Cr^{2-\vep},\quad 0<r<r_0,
\ee
as claimed.

In view of (\ref{2.13}), we deduce that
\be 
\int_0^\infty (w_r)^2\,\dd r<\infty.
\ee
Therefore, there is a sequence $\{r_n\}$, $r_n\to\infty$ as $n\to\infty$, such that 
\be \label{2.25}
\lim_{n\to\infty}w_r(r_n)=0.
\ee
 Furthermore, since $G$ in (\ref{ODE1}) is of compact support, there is
some $K>0$ such that $G(r)=0$ for $r> K$. Thus, using $w\leq0$ and the definition (\ref{4}), we see that $w_{rr}$ satisfies
\be  \label{2.26}
|w_{rr}(r)|=\frac2{{\Xi}^2} |\e^w-1|\leq\frac{64}{S^2}\e^{-\frac{4r}S},\quad r\geq K_0=\max\left\{K,\frac S2\ln 2\right\}.
\ee
Integrating (\ref{ODE1}), using (\ref{2.26}), and applying (\ref{2.25}), we arrive at
\be\label{2.27}
0\leq w_r(r)\leq \frac{16}S \e^{-\frac{4r}S},\quad r>K_0.
\ee
Integrating (\ref{2.27}), we see that there is some number $w_\infty\leq0$ such that
\be 
\lim_{r\to\infty} w(r)=w_\infty, \quad 0\leq w_\infty -w(r)\leq 4\e^{-\frac{4r}S},\quad r>K_0,
\ee
as anticipated.

To see the uniqueness of $w_\infty$, we assume there are solutions of (\ref{ODE1}), say $W_1$ and $W_2$ such that
\be 
W_1(\infty)=W_{1,\infty},\quad W_2(\infty)=W_{2,\infty},\quad W_{1,\infty}>W_{2,\infty}.
\ee
Let $W=W_1-W_2$. Then $W$ satisfies 
\be \label{Weq}
W_{rr}=\frac2{{\Xi}^2}\e^{\xi}\, W,\quad r>0,
\ee
where $\xi$ lies between $W_1$ and $W_2$. Since $W(0)=0$ and $W(\infty)>0$, we have $W(r)>0$ for all $r>0$. Otherwise, let $W(r_0)\leq0$ for some $r_0>0$. We may assume that $W$ attains its
global minimum at $r_0$. If $W(r_0)=0$, then $W_r(r_0)=0$, which implies $W\equiv0$ by the uniqueness of solution to the initial value problem of an ordinary differential equation, contradicting
the condition $W(\infty)>0$. So $W(r_0)<0$ but this contradicts the fact $W_{rr}(r_0)\geq0$. Hence $W(r)>0$ for all $r>0$.

Since $W(\infty)$ is finite, there is a sequence $\{r_n\}$, $r_n\to\infty$ as $n\to\infty$ such that $W_r(r_n)\to0$ as $n\to\infty$. Integrating (\ref{Weq}) over $(r,r_n)$ and letting $n\to\infty$, we have
\be 
W_r(r)=-\int_r^\infty\frac2{{\Xi}^2}\e^{\xi}\, W\,\dd\rho,\quad r>0.
\ee
Using $W(0)=0$, $W>0$, and the above, we see that $W(r)$ decreases. In particular, $W(\infty)<0$, which is another contradiction.

The proof of the lemma is complete.
\end{proof}

Despite of the above uniqueness result, we are unable to show that $w_\infty=0$. 

We are now ready to solve (\ref{2.2}). We can state

\begin{theorem}\label{theorem} The equation (\ref{2.2}) has a bounded solution $v$ satisfying $v(r,t)={\rm\mbox{O}}(r^{2-\vep})$ as $r\to0$ where $\vep>0$ is arbitrarily small.
\end{theorem}
\begin{proof} Let $w$ be the solution of (\ref{ODE1}) stated in Lemma \ref{lemma2.5}. Since $u_0\leq0$, we have
\be 
\Delta w\geq \frac2{{\Xi}^2}(\e^{u_0+w}-1)+g(r,t).
\ee
In other words, $v^-=w$ is a lower solution of the equation (\ref{2.2}). Combining with Lemma \ref{lemma0}, we have $v^+\geq 0\geq v^-$. Using elliptic method, we get a solution
$v$ of (\ref{2.2}) satisfying $v^-\leq v\leq v^+$. Since $v^+=-u_0$ is of compact support in $(0,\infty)\times(0,\beta)$, we obtain from Lemma \ref{lemma2.5} that $v$ is bounded
and satisfies $v(r,t)=\mbox{O}(r^{2-\vep})$ as $r\to0$ for any small number $\vep>0$.
\end{proof}

\section{Calculation of topological charge}
\setcounter{equation}{0}

Let $v=v(r,t)$ be the $\beta$-periodic solution of (\ref{2.2}) obtained in Theorem \ref{theorem}. Define the $\beta$-averaged function by
\be 
\bar{v}(r)=\frac1\beta\int_0^\beta v(r,t)\,\dd t.
\ee
From the uniform decay estimate $v(r,t)=\mbox{O}(r^{2-\vep})$ (when $r>0$ is small), $\bar{v}(0)=0$, and the L'Hopital's rule, we have
\be \label{vr1}
\lim_{r\to0}\bar{v}_r(r)=\lim_{r\to0}\frac{\bar{v}(r)}r=\frac1\beta\lim_{r\to0}\int_0^\beta \frac1r v(r,t)\,\dd t=0.
\ee
On the other hand, since $\bar{v}$ is bounded, there is a sequence $\{r_n\}$, $r_n\to\infty$ as $n\to\infty$, such that
$\bar{v}_r(r_n)\to0$ as $n\to\infty$. Integrating (\ref{2.2}) over $(r,r_n)\times(0,\beta)$ and letting
$n\to\infty$, we have
\bea \label{vr2}
\bar{v}_r(r)&=&\frac1\beta\int_r^{\infty}\int_0^\beta\frac2{{\Xi}^2}(1-\e^{v+u_0})\,\dd t\,\dd \rho\nn\\
&\leq&\int_r^\infty \frac{64}{S^2}\e^{-\frac{4\rho}S}\,\dd\rho\to0\quad\mbox{as }r\to\infty,
\eea
where we have used the property $v+u_0\leq0$.

Integrating $\Delta v$ over $0<r<\infty, 0<t<\beta$ and applying (\ref{vr1}) and (\ref{vr2}), we have
\be \label{ddv}
\int_0^\infty\int_0^\beta \Delta v\,\dd t\,\dd r=\beta\left(\lim_{r\to\infty}\bar{v}_r(r)-\lim_{r\to0}\bar{v}_r(r)
\right)=0.
\ee

On the other hand, recall that $u_0=\eta U_0$ is of compact support and $U_0$ may be decomposed as the sum
of a regular and singular parts in the form
\be \label{U0}
U_0(x)=-\sum_{j=1}^N\ln(1+|x-p_j|^2) +\sum_{j=1}^N\ln|x-p_j|^2\equiv {\tilde R}(x)+{\tilde S}(x). 
\ee
Thus, with ${\cal M}=\{(r,t)\,|\,0<r<\infty,0\leq t\leq \beta\}$, we insert (\ref{U0}) to have
\bea\label{g}
\int_{\cal M} g(r,t)\,\dd r\,\dd t&=&\lim_{\vep\to0}\int_{{\cal M}\setminus\cup_{j=1}^N\{x\,|\,|x-p_j|\geq\vep\}}
(-\Delta u_0)\,\dd r\,\dd t\nn\\
&=&-\lim_{\vep\to0}\int_{{\cal M}\setminus\cup_{j=1}^N\{x\,|\,|x-p_j|\geq\vep\}}
\Delta (\eta {\tilde S})\,\dd r\,\dd t\nn\\
&=&\lim_{\vep\to0}\sum_{j=1}^N\oint_{|x-p_j|=\vep}\frac{\pa (\eta {\tilde S})}{\pa n}\,\dd S_\vep=4\pi N,
\eea
where $\dd S_\vep$ is the line element and $\pa/\pa n$ denotes the outnormal derivative on the circle $|x-p_j|=\vep$
($j=1,\cdots,N$).

Integrating (\ref{2.2}) over ${\cal M}$ and using (\ref{ddv}) and (\ref{g}), we obtain
\be 
\int_0^\infty\int_0^\beta\frac1{{\Xi}^2}(1-\e^{v+u_0})\,\dd t\,\dd r=2\pi N.
\ee
This result is important because through the relation $|\phi|^2=\e^u=\e^{v+u_0}$ and (\ref{2}), we arrive at the
anticipated flux quantization condition
\be \label{flux}
\Phi=\int_{\cal M} F_{tr}\,\dd t\,\dd r=\int_{\cal M} \frac1{{\Xi}^2}(1-|\phi|^2)\,\dd t\,\dd r=2\pi N.
\ee

In order to calculate the total action, we need to establish some suitable decay estimates for the gradient of
a solution obtained, which will be considered in the next section.
%In conclusion, we obtain the following main existence theorem.

%\begin{theorem} For any integer $N\geq1$, there is an instanton solution of the second
%Chern number $N$ over the hyperbolic space which may be realized by prescribing $N$ point vortices in the cylinder
%$0<r<\infty, 0<t<\beta$ so that the field configurations are $\beta$-periodic in the temporal coordinate $t$.
%\end{theorem}

\section{Decay estimates for the gradient of solution}
\setcounter{equation}{0}

In order to compute the action of a multiple instanton, we need to derive suitable decay estimates for the first derivatives
of the solution $v$ of (\ref{2.2}) obtained earlier. 

We shall first consider the decay estimates near $r=0$. Since both $u_0(r,t)$ and $g(r,t)$ are compactly supported in ${\cal M}$, we may
choose $r_0>0$ sufficiently small so that the supports of $u_0$ and $g$ are contained in $\{(r,t)\in{\cal M}\,|\,r> 2r_0\}$ (say). Hence $v$ satisfies
\be \label{3.1}
\Delta v=\frac2{{\Xi}^2}(\e^v-1),\quad 0<r<2r_0.
\ee

\begin{lemma}\label{lemma3.1} Let $v$ be the solution of (\ref{2.2}) obtained earlier. For any arbitrarily small number $\vep>0$, there is a constant $C(\vep)>0$ independent of
$r,t$ such that the estimates
\be\label{3.2}
|v_t|\leq C(\vep) r^{2-\vep},\quad - C(\vep) r^{1-\vep}\leq  v_r\leq C(\vep) r^{2-\vep},\quad 0<r<r_0,
\ee
hold. In particular, $|\nabla v|\to0$ uniformly as $r\to0$.
\end{lemma}
\begin{proof} Since ${\Xi}(r)=\mbox{O}(r)$, $v(r,t)=\mbox{O}(r^{2-\vep})$ uniformly as $r\to0$ for $\vep>0$ arbitrarily small, and $v\leq0$, we see that $\Delta v\in L^p ({\cal M}_{2r_0})$
for any $p>2$ (say) where ${\cal M}_{\delta}=\{(r,t)\in{\cal M}\,|\, 0<r<\delta\}$. Elliptic $L^p$-estimates indicate that $v\in W^{2,p}({\cal M}_{2r_0})$.
Using the embedding $W^{2,p}({\cal M}_{2r_0})\to C^1(\overline{{\cal M}_{2r_0}})$, we see that $|\nabla v|$ is bounded over ${\cal M}_{2r_0}$.

For any $h>0$, consider the function
\be 
v^h(r,t)=\frac{v(r,t+h)-v(r,t)}h.
\ee
Then $\{v^h\}$ is uniformly bounded over ${\cal M}_{r_0}$ and $v^h(r,t)\to0$ as $r\to0$. Of course, in view of (\ref{3.1}), $v^h$ satisfies the equation
\be 
\Delta v^h=\frac2{{\Xi}^2}\e^{u^h} v^h,\quad 0<r<r_0,
\ee
where $u^h(r,t)$ lies between $v(r,t)$ and $v(r,t+h)$. Using the comparison function $W$ for fixed $\vep>0$ defined in (\ref{W}) and applying (\ref{V0}), we get
\bea \label{3.5}
\Delta(v^h+W)&=&\frac2{{\Xi}^2} \e^{u^h} v^h+(2-\vep)(1-\vep) r^{-2} W\nn\\
&\leq& r^{-2}\frac2{\cosh^2\left(\frac{2\xi(r)}S\right)} \e^{u^h} (v^h+W),\quad 0<r<r_0,
\eea
where we have assumed that $r_0>0$ is sufficiently small and applied the uniform limit $v(r,t+h)\to0$ as $r\to0$. Since $\{v^h\}$ is bounded, we may also assume that $C>0$
in (\ref{W}) is large enough so that
\be \label{3.6}
v^h(t,r_0)+W(r_0)\geq0\quad\mbox{for all }t.
\ee
Thus, the boundary condition consisting of $v^h+W=0$ at $r=0$ and (\ref{3.6}), the inequality (\ref{3.5}), and the maximum principle together lead us to
\be 
v^h(r,t)+W(r)\geq0,\quad 0<r<r_0.
\ee
Similarly, we have
\be 
\Delta(v^h-W)\geq r^{-2}\frac2{\cosh^2\left(\frac{2\xi(r)}S\right)} \e^{u^h} (v^h-W),\quad 0<r<r_0,
\ee
and we deduce $v^h-W\leq0$, $0<r<r_0$. Summarizing these results, we arrive at $|v^h(r,t)|\leq W(r)$, $0<r<r_0$. Letting $h\to0$, we obtain $|v_t(r,t)|\leq W$,
$0<r<r_0$, as stated in (\ref{3.2}).

In order to get the decay estimate for $v_r$ as $r\to0$, we note that
\be 
\lim_{r\to0}v_r(r,t)=\lim_{r\to0}\frac{v(r,t)}r=0,
\ee
by virtue of Theorem \ref{theorem}. Differentiating (\ref{3.1}), we find, using $v\leq0$, the inequality
\be \label{3.10}
\frac2{{\Xi}^2}\e^v v_r-\frac4{{\Xi}^3}(\e^v-1)\cosh\left(\frac{2r}S\right)=\Delta v_r\geq\frac2{{\Xi}^2}\e^v v_r,\quad0<r<r_0.
\ee
Consequently, we have
\be 
K(r) v_r\leq r^2\Delta v_r\leq K(r) v_r+C_1 r^{1-\vep},\quad 0<r<r_0,
\ee
where $C_1>0$ is an absolute constant and $K=r^2(2/{\Xi}^2)\e^v$ satisfies
\be 
\lim_{r\to0} K(r)=2.
\ee
Hence, for the function $W$ defined in (\ref{W}), we have
\be \label{3.13}
r^2\Delta (v_r-W)\geq K(r) (v_r-W),\quad 0<r<r_0,
\ee
and 
\bea \label{3.14}
r^2\Delta(v_r+W_r)&\leq&K(r) v_r+C_1 r^{1-\vep}-\vep(1-\vep) W_r\nn\\
&\leq &K(r)v_r+\frac{C_1}C W_r,\quad 0<r<r_0,
\eea
where $0<\vep<1$ and $C$ is as given in (\ref{W}). We may choose $C$ large enough so that $K(r)\geq C_1/C$, $0<r<r_0$. Then (\ref{3.14}) gives us
\be \label{3.15}
r^2\Delta(v_r+W_r)\leq K(r)(v_r+W_r),\quad 0<r<r_0.
\ee
Using the same maximum principle argument in (\ref{3.13}) and (\ref{3.15}) as before, we see that, when $C>0$ in (\ref{W}) is large enough, we have
\be 
-W_r\leq v_r \leq W,\quad 0<r<r_0,
\ee
which establishes the decay estimate for $v_r$ stated in (\ref{3.2}).
\end{proof}

 We now consider the decay estimate for $|\nabla v|$ as $r\to\infty$. Since we do not know whether $v\to0$ as $r\to\infty$, we encounter a somewhat delicate
situation that $v$ may not lie in $L^2({\cal M})$.

Similar to (\ref{3.1}), we know that $v$ satisfies
\be \label{3.16}
\Delta v=\frac2{{\Xi}^2}(\e^v-1),\quad r>\delta,
\ee
where $\delta>0$ is sufficiently large. For convenience, we set ${\cal M}^\delta=\{(r,t)\in{\cal M}\,|\,r>\delta\}$. 

\begin{lemma}\label{lemma3.2} We have $|\nabla v|\in L^2({\cal M}^\delta)$.\end{lemma}

\begin{proof} We may extend $v$ outside ${\cal M}^\delta$ smoothly to get a new function, say $w$, so that $w=0$ for $r<\delta/2$ (say). Hence $w$ satisfies
\be \label{3.17}
\Delta w=\frac2{{\Xi}^2}(\e^w-1)+h(r,t),\quad (r,t)\in {\cal M},
\ee
where $h$ is of compact support and smooth. Choose a smooth function $\eta(r)$ in $r\geq0$ such that
\be \label{3.19}
0\leq\eta\leq1;\quad \eta(r)=1,\quad 0\leq r\leq1;\quad \eta(r)=0,\quad r\geq 2.
\ee
Define $\eta_\rho=\eta(r/\rho)$ for $\rho>0$. Multiplying (\ref{3.17}) by $\eta_\rho^2 w$ and integrating, we have
\bea 
\int_{\cal M}\left|\frac2{{\Xi}^2}(\e^w-1)+h\right|\eta^2_\rho |w|\,\dd r\dd t&\geq&\int_{\cal M}\eta^2_\rho|\nabla w|^2\,\dd r\dd t-2\int_{\cal M}\eta_\rho|\nabla w||w||\nabla\eta_\rho|\,\dd r\dd t\nn\\
&\geq&\frac12\int_{\cal M}\eta^2_\rho|\nabla w|^2\,\dd r\dd t-2\int_{\cal M} w^2|\nabla\eta_\rho|^2\,\dd r\dd t\nn\\
&\geq&\frac12\int_{\cal M}\eta^2_\rho|\nabla w|^2\,\dd r\dd t-\frac C\rho,\label{3.20}
\eea
where $C>0$ is a constant depending on $|w|_\infty$ and $\beta$ only. Letting $\rho\to\infty$ in (\ref{3.20}) and recalling (\ref{3.19}), we see that $|\nabla w|\in L^2({\cal M})$ and the lemma follows.
\end{proof}

\begin{lemma}\label{lemma3.3} There is constant $C_\delta>0$ such that
\be 
\sup_{{\cal M}^\delta}\{|\nabla v|(r,t)\}\leq C_\delta,
\ee
and $|\nabla v|(r,t)\to 0$ as $r\to\infty$.
\end{lemma}

\begin{proof} Differentiating (\ref{3.16}), we have
\be \label{3.22}
\Delta v_t=\frac2{{\Xi}^2}\e^v v_t,\quad r>\delta.
\ee
Using Lemma \ref{lemma3.2} and elliptic theory, we see that $v_t\in W^{2,2}({\cal M}^\delta)$. Hence $v_t$ is bounded and $v_t\to0$ as $r\to \infty$.

Similarly, differentiating (\ref{3.16}) with respect to $r$, we have
\be \label{3.23}
\Delta v_r=\frac2{{\Xi}^2}\e^v v_r-\frac4{{\Xi}^3}(\e^v-1)\cosh\left(\frac{2r}S\right),\quad r>\delta,
\ee
whose right-hand side lies in $L^2({\cal M}^\delta)$. Consequently, $v_r$ is bounded and $v_r\to0$ as $r\to\infty$, which establishes the lemma.
\end{proof}

The gradient decay estimates for the solution near the boundary of the domain $\cal M$ allows us to compute the action in terms of the topological invariant
explicitly, which will be carried out in the next section.

\section{Calculation of action}
\setcounter{equation}{0}

Following \cite{Harland}, it will be convenient to express the dimensionally reduced action (\ref{action-phi-a}) of the gauge field $A$ in terms of the $(r,t)$-coordinates explicitly as
\be \label{S}
S(A)=\frac\pi2\int_{\cal M}\left({\Xi}^2 F^2_{tr}+\left[\frac{1-|\phi|^2}{\Xi}\right]^2+2|D_t\phi|^2+2|D_r\phi|^2\right)\,\dd r\dd t.
\ee
It can be checked that the useful identities
\bea
|D_t\phi|^2+|D_r\phi|^2&=&\ii(D_r\phi\overline{D_t\phi}-\overline{D_r\phi}D_t\phi)+|D_r\phi+\ii D_t\phi|^2,  \label{id1}\\
\ii(D_t\phi\overline{D_r\phi}-\overline{D_t\phi}D_r\phi)&=&\ii(\pa_t[\phi\overline{D_r\phi}]-\pa_r[\phi\overline{D_t\phi}])-F_{tr}|\phi|^2,\label{id2}
\eea
hold. Inserting (\ref{id1}) and (\ref{id2}) into (\ref{S}) and applying (\ref{1}) and (\ref{2}), we have
\bea \label{SA}
S(A)&=&\frac\pi2\int_{\cal M}\bigg(\bigg[{\Xi}F_{tr}-\frac1{\Xi}(1-|\phi|^2)\bigg]^2+2F_{tr}(1-|\phi|^2)\nn\\
&&+2|D_r\phi+\ii D_t\phi|^2-2\ii (D_t\phi\overline{D_r\phi}-\overline{D_t\phi}D_r\phi)\bigg)\,\dd r\dd t\nn\\
&=&\pi\int_{\cal M} F_{tr}\,\dd x-\ii \pi \int_{\cal M} (\pa_t[\phi\overline{D_r\phi}]-\pa_r[\phi\overline{D_t\phi}])\,\dd r\dd t.
\eea

By virtue of (\ref{DD}), Lemma \ref{lemma3.1}, and Lemma \ref{lemma3.3}, we see that the last integral on the right-hand side of (\ref{SA}) vanishes. Therefore we obtain the quantized
minimum action
\be 
S(A)=2\pi^2 N,
\ee
as a consequence of the flux formula (\ref{flux}).

\medskip

In summary, we have seen that our main existence theorem for hyperbolic calorons of arbitrary scalar curvature, time period, and topological charge stated in Section 2 is established 
 in Section 4 through 
a construction of the solution of the multivortex equation derived by Harland \cite{Harland}, a computation of the associated topological charge in Section 6, and a calculation of the dimensionally reduced action 
in Section 7, which is based on the gradient boundary estimates of the solution obtained in Section 6.
\medskip 

{\bf Data accessibility.} This work does not have any experimental data.

\medskip 

{\bf Competing interests.} We do not have competing interests.

\medskip 

{\bf Authors' contributions.} LS and RS offered ideas and insights in the mathematical formulation and conception of the problem. YY developed analytic
methods to tackle the problem and wrote the paper. All authors gave final approval for publication.
\medskip 

{\bf Funding.} 
The research of YY was partially supported
by National Natural Science Foundation of China under Grant No. 11471100.

\medskip
\medskip

\small{

}

\end{document}